\documentclass{llncs}
\usepackage{amsmath}
\usepackage{amssymb}

\newcommand{\md}{\ensuremath \textrm{md}}
\newcommand{\dd}{\ensuremath \textrm{dd}}
\newcommand{\palto}{type}
\newcommand{\M}{\ensuremath{\mathcal M}}

\begin{document}

\title{Parameterized Modal Satisfiability}

\author{Antonis Achilleos \and Michael Lampis \and Valia Mitsou}

\institute{Computer Science Department,  \newline Graduate Center, City University of New York, \newline 365 5th Ave New York, NY 10016 USA
\email{antach@corelab.ntua.gr,mlampis@gc.cuny.edu,vmitsou@cs.gc.cuny.edu} }
%\email{mlampis@gc.cuny.edu}
%\email{vmitsou@cs.gc.cuny.edu}

\maketitle

\begin{abstract}
We investigate the parameterized computational complexity of the
satisfiability problem for modal logic and attempt to pinpoint
relevant structural parameters which cause the problem's
combinatorial explosion, beyond the number of propositional
variables $v$. To this end we study the modality depth, a natural
measure which has appeared in the literature, and show that, even
though modal satisfiability parameterized by $v$ and the modality
depth is FPT, the running time's dependence on the parameters is a
tower of exponentials (unless P=NP). To overcome this limitation we
propose several possible alternative parameters, namely diamond
dimension, box dimension and modal width. We show fixed-parameter
tractability results using these measures where the exponential
dependence on the parameters is much milder than in the case of
modality depth thus leading to FPT algorithms for modal
satisfiability with much more reasonable running times.

\end{abstract}

\section{Introduction}

In this paper we consider the computational complexity of deciding
two fundamental logic problems, namely formula satisfiability and
formula validity, for modal logics, focusing on the standard modal
logic K. We attempt to present a new point of view on this important
topic by making use of the parameterized complexity framework, which
was pioneered by Downey and Fellows. Although the complexity of
satisfiability for modal logic has been studied extensively in the
past, to the best of our knowledge this is the first time this has
been done from an explicitly parameterized perspective. Moreover,
the parameterized complexity of logic problems has been a fruitful
field of research and we hope to extend this success to modal logic
(some examples are the celebrated theorem of Courcelle
\cite{DBLP:journals/iandc/Courcelle90} or the results of
\cite{DBLP:journals/apal/FrickG04}; for an excellent survey on the interplay
between logic, graph problems and parameterized complexity see
\cite{DBLP:journals/eccc/Grohe07}).

Modal logic is a family of systems of formal logic where the truth value of a sentence $\phi$ can be qualified by
modality operators, usually denoted by $\Box$ and $\Diamond$. Depending on the specific modal logic and the
application one considers, $\Box\phi$ and $\Diamond\phi$ can be informally read to mean, for example, ``it is
necessary that $\phi$'', or ``it is known that $\phi$'' for $\Box$ and ``it is possible that $\phi$'' for
$\Diamond$. The fundamental normal modal logic system is known as K, while other common variations of this logic
system include T,D,S4,S5. Modal logic systems provide a diverse universe of logics able to fit many modern
applications in computer science (for example in AI or in game theory), making modal logic a widespread topic of
research. The interested reader in the recent state of modal logic and its applications is directed to
\cite{handbook-modal}.

As in propositional logic, the satisfiability problem for modal logic is one of the most important and fundamental
problems considered and many results are known about its (traditional) computational complexity. Ladner in
\cite{lad} showed that satisfiability for K, T and S4 is PSPACE-complete, while for S5 the problem is NP-complete.
%, thus settling the issue of determining
%the complexity of the problem for the most well-known normal modal
%logics.
Furthermore,
% \^{S}vejdar in \cite{sve} showed that the
%problem remains PSPACE-complete for GL, even if we consider formulae
%with only one variable and
in \cite{DBLP:conf/aiml/ChagrovR02}
%the same is shown for S4 and Grz, while
it is shown that satisfiability for K and K4 is PSPACE-complete even
for formulae without any variables.
%In \cite{halreg},
%Halpern and R\^{e}go showed that the negative introspection axiom is in an
%essential way what makes the difference between normal modal logics whose
%satisfiability problem is in NP and those for which it is PSPACE-complete.
It should be noted that the satisfiability of propositional logic is a subcase of satisfiability for any normal
modal logic, thus for any normal modal logic the problem is NP-hard.  Other results are known for multimodal
logics; all of the above are PSPACE-complete in the multimodal case.  In this paper we will focus on normal
monomodal logics and mainly on K. For an introduction to modal logic and its complexity see
\cite{130913,Fagin1995ReasoningAboutKnowledge}.

Traditional computational complexity theory attempts to characterize the complexity of a problem as a function of
the input size $n$. The notion of parameterized complexity introduces to every hard problem a structural parameter
$k$, which attempts to capture the aspect of the problem which causes its intractability. The central notion of
tractability in this theory is called fixed-parameter tractability (FPT): an algorithm is called FPT if it runs in
time $O(f(k)\cdot n^c)$, where $f$ is any recursive function and $c$ a constant. For an introduction to the vast
area of parameterized complexity see \cite{downey1999pc,flum2006pct}.

Because the definition of FPT allows for any recursive function $f(k)$, fixed-parameter tractable problems can
have complexities which depend on $k$ in very different ways, ranging from sub-exponential to non-elementary.
Thus, it is one of the main goals of parameterized complexity research to find the best possible $f(k)$ for every
problem and this will be one of the main concerns of our work.

%\vspace{0.1cm}

\textit{\textbf{Our contribution}}
In this paper we study the complexity of modal satisfiability and
validity from a parameterized, or multi-variate, point of view. Just
as parameterized complexity attempts to refine traditional
complexity theory by more specifically identifying the aspects of an
intractable problem which cause the problem's unavoidable
combinatorial explosion, we attempt to identify some structural aspects of
modal formulae which can have an impact on the solvability of
satisfiability.

One natural parameter for the satisfiability problem (in any logic)
is the number of propositional letters in the formula, which we
denote by $v$. In propositional logic, when $v$ is taken as a
parameter, the propositional satisfiability problem trivially
becomes fixed-parameter tractable.  As was already mentioned, this
does not generally hold in the case of satisfiability for modal
logics where the problem is hard even for constant number of
variables.

On the other hand since the satisfiability problem for modal logics
is a generalization of the same problem for propositional logics,
considering the modal satisfiability problem without bounding the
number of variables or imposing some other propositional restriction
on the formulae will result in an intractable problem. Although it
would be interesting to investigate modal satisfiability when
certain structural propositional restrictions are placed (for
example, we could say we are interested in formulae such that
removing all modality symbols leaves a 2-CNF or a Horn formula,
which are tractable cases of propositional satisfiability) this goes
beyond the scope of this work\footnote{However, see
\cite{nguyen2005cfm} for related (non-parameterized) complexity
results}. In this paper we will focus on strictly modal structural
formula restrictions and therefore we will assume that the best way
to make propositional satisfiability tractable is to restrict the
number of variables. For our purposes the conclusion is that for
modal satisfiability to become tractable, bounding $v$ is necessary
but not sufficient.

Motivated by the above we take the approach of a double
parameterization: we investigate the complexity of satisfiability
and validity when $v$ is considered a parameter and at the same time
some other aspect contributing to the problem's complexity is
identified and bounded.

We first study a natural notion of formula complexity called
modality depth or modal depth. This complexity measure was already
known in \cite{DBLP:journals/ai/Halpern95} where in fact a
fixed-parameter tractability result was shown when the problem is
parameterized by the sum of $v$ and the modality depth of the
formula. However, since parameterized complexity was not well-known
at the time, in \cite{DBLP:journals/ai/Halpern95} it is only pointed
out that the problem is solvable in linear time for fixed values of
the parameters, without mentioning how different values of $v$ and
the depth affect the running time.  We address this by upper
bounding the running time by an exponential tower of height equal to
the modality depth of the formula. More importantly, we show a lower
bound argument which proves that even though the problem is FPT, this
exponential tower in the running time cannot be avoided unless P=NP
(Theorem~\ref{thm:lower}).  Our hardness proof follows an approach of encoding
a propositional formula into a modal formula with very small modality depth.
This draws a nice connection with previously known lower bound results of this
form which also use a similar idea to prove the hardness of some (non-modal)
model checking  problems for first and second-order logic
(\cite{DBLP:journals/apal/FrickG04} and the relevant chapter in
\cite{flum2006pct}).

This result indicates that modal depth is unlikely to be a very useful
parameter because even for formulae where the depth is very moderate the
satisfiability problem is still very hard. This begs the natural question of
whether there is a way to work around the lower bound of
Theorem~\ref{thm:lower} by using another formula complexity measure in the
place of modal depth.  It is worth noting that this is a major difference
between Theorem~\ref{thm:lower} and the results of
\cite{DBLP:journals/apal/FrickG04} on FO and MSO model checking on trees,
because in that case the lower bound applies to the problem parameterized by
the formula \emph{size}, not its quantifier depth.  Since a natural formula
complexity measure would likely be bounded by some simple function of the
formula size, a search for good formula parameters is very unlikely to bear
fruit in that case. However, we show that the modal satisfiability case is
quite different: we define and study several new notions of modal formula
complexity and show that unlike modality depth, these notions can be used to
obtain not only fixed-parameter tractability results but also much more
reasonable running times.

Specifically, we define the notions of diamond dimension and box
dimension of a modal formula and show that satisfiability is FPT
when parameterized by $v$ and the diamond dimension and validity is
FPT when parameterized by $v$ and box dimension and that in both
cases the running times are doubly exponential in the parameters.
Then we define a measure called modal width and show that both
satisfiability and validity are FPT when parameterized by $v$ and
the modal width and the dependence on the parameters is singly
exponential. Thus, our work points out that trying to solve satisfiability for
formulae where our proposed measures has a moderate value can be done much more
efficiently than by using the already known modality depth. All our work
focuses on K, but many of our results easily carry over to other modal logics
without much modification.

%\vspace{0.1cm}

\textit{\textbf{Notation}}
The modal language of logic $K$ contains exactly the formulae that can be
constructed using the standard propositional operators ($\land,\lor,\neg$) and
the unary modality operators ($\Box,\Diamond$).  Standard Kripke semantics are
considered here: a Kripke structure is a set of states $W$, an accessibility
relation $R$ between states and a valuation of propositional letters in each
state. A modal formula's truth value in a state is defined in the usual way, as
in propositional logic, with the addition of $\Box \phi$ being true in $s$ iff
$\phi$ is true in every accessible state.  $\Diamond \phi$ is usually
considered short for $\neg \Box \neg \phi$.  We implicitly assume that our
language includes the constants $\bot$ and $\top$, for false and true, but
these too can also be considered shorthand for $x\land \neg x$ and $x\lor \neg
x$ respectively. When a formula $\phi$ is true (satisfied) in a state $s$ of a Kripke
structure $\mathcal{M}$ we write $(\mathcal{M},s) \models \phi $. A formula $\phi$ is said to be satisfiable if there exists a Kripke structure $\mathcal{M}$ and a state $s$ of that structure that satisfy the formula.
A formula $\phi$ is said to be valid if any Kripke structure $\mathcal{M}$ and state $s$ of that structure satisfy the formula.

%The rest of this paper is structured as follows:
%Section~\ref{sec:depth} contains the results concerning modal depth,
%Section~\ref{sec:dimension} the results for diamond and box
%dimension and Section~\ref{sec:width} the results concerning modal
%width.

\section{Modal Depth} \label{sec:depth}

In this Section we give the definition of modality depth. As we will see, a
fixed-parameter tractability result can be obtained when satisfiability is
parameterized by $v$ and the modality depth of the input formula.  This was
first observed in \cite{DBLP:journals/ai/Halpern95}, but in this section we
more precisely bound the running time (in \cite{DBLP:journals/ai/Halpern95} it
was simply noted that the running time is linear for constant depth and
constant $v$ with a hidden constant which ``may be huge''). More importantly we
show that the ``huge constant'' cannot be significantly improved by giving a
hardness proof which shows that, if the running time of an algorithm for modal
satisfiability is significantly less than an exponential tower of height equal
to the modality depth, then P=NP.

\begin{definition}
The modality depth of a modal formula $\phi$ is defined inductively
as follows:
\begin{itemize}
\item $\md(p)=0 $, if $p$ is a propositional letter,
\item $\md(\Diamond \phi) =\md(\Box \phi) = 1+\md(\phi)$,
\item $\md(\phi_1\lor \phi_2) = \md(\phi_1\land \phi_2) = \max\{\md(\phi_1),\md(\phi_2)\}$,
\item $\md(\neg \phi) = \md(\phi)$
\end{itemize}
\end{definition}

\begin{theorem} (\cite{DBLP:journals/ai/Halpern95}) \label{thm:depth_upper}
Modal satisfiability and modal validity for the logic $K$ are FPT when parameterized by $v$ and $\md(\phi)$.
\end{theorem}

\begin{proof}
We define the $d$-\palto\ of a state $s$ in a Kripke structure $\mathcal M$ to
be the set $\{\phi\ |\ (\M,s) \models \phi \textrm{ and } \md(\phi) \leq d \}$.
We will prove by induction on $d$ that if we restrict ourselves to formulae
with at most $v$ variables then for any $d \geq 0$ there are at most $f_v(d)$
$d$-\palto s, where $f_v$ is the function recursively defined: $ f_v(0) = 2^v$,
\hbox{$f_v(n+1) = 2^{f_v(n) + v}$}.

\begin{description}
\item[For $d = 0$] If $\md(\phi) = 0$, then the formula is propositional, thus
the $0$-\palto\ of any state is directly defined by the set of propositional
letters assigned true in the state. The number of all such possible sets of
variables is $2^v = f_v(0)$.
\item[For the case of $d+1$] The $(d+1)$-\palto\ of
a state $s$ depends on the assignment of the propositional letters in $s$ and
on the truth values of formulae of the forms $\Box \phi'$ and $\Diamond \phi'$,
where $\md(\phi')\leq d$. Notice that these truth values depend only on the set
of $d$-\palto s of the accessible states from $s$. Thus the number of different
$(d+1)$-\palto s on a state $s$ is $f_v(d+1) = 2^{f_v(d)+v}$.
\end{description}

Now, suppose that $\phi$ is a satisfiable formula of modality depth $d\ge 1$.
We will show how to construct a Kripke structure of about $f_v(d-1)$ states to
satisfy $\phi$. To achieve this, for all $i\in\{0,1,\ldots,d-1\}$ and for all
$i$-types we will construct a state of that $i$-type, thus in total we will
construct $\sum_{i=0}^{d-1} f_v(i)= O(f_v(d-1))$ states. To construct the
$f_v(0)=2^v$ states that give all the different $0$-types we just construct
$2^v$ states, each with a different valuation of the propositional variables.
For the subsequent levels, to construct all the states for all the different
$(i+1)$-types we pick for each state a set of successor states out of the
states that give us the different $i$-types and a valuation of the
propositional variables. If $\phi$ is satisfiable, it must be satisfiable in
this model by adding a new state $s$, selecting a subset of the states that
give us the different $(d-1)$-types to be its successors and a valuation of the
propositional variables in $s$. The number of combinations of all possible
subsets of successors and all variable valuations is $f_v(d)$, so the problem
is solvable in $O(f_v(d)\cdot f_v^2(d-1)\cdot|\phi|)$, because the model has
$O(f_v(d-1))$ states and thus size $O(f_v^2(d-1))$ and model checking can be
performed in bilinear time.

Since modal validity is the dual problem of modal satisfiability and negating
the formula doesn't change its modality depth, the same results hold for this
problem too.

\end{proof}
%The running time of the algorithm of
%\cite{DBLP:journals/ai/Halpern95} is $O(f_v(\md(\phi))|\phi|)$ where
%$f_v(d)$ is the function recursively defined as $f_v(0)=2^v$ and
%$f_v(d+1)=2^{f_v(d)+v}$ (the full proof of this fact is omitted due
%to space constraints).

%formula mouriki!!!
Let us now proceed to the main result of this Section, which is that even
though modal satisfiability is fixed-parameter tractable, the exponential tower
in the running time cannot be avoided. Specifically, we will show that solving
modal satisfiability parameterized by modality depth, even for constant $v$, 
requires a running time which is a tower of exponentials with height depending
on the modality depth. We will prove this under the assumption that P$\neq$NP,
by reducing the problem of propositional satisfiability to our problem.

Suppose that we are given a propositional CNF formula $\phi_p$ with variables
$x_1,\ldots,x_n$ and we need to check whether there exists a satisfying
assignment for it.  We will encode $\phi_p$ into a modal formula with small
depth and a constant number of variables. In order to do so we inductively
define a sequence of modal formulae.

\newcommand{\LL}{\mathcal{L}} \newcommand{\CC}{\mathcal{C}}
\newcommand{\FF}{\mathcal{F}} \newcommand{\SSS}{\mathcal{S}}
\newcommand{\CCC}{\mathcal{CA}}

\begin{itemize}

\item In order to encode the variables of $\phi_p$ we need some formulae to
encode numbers(the indices of the variables). The modal formula
$v_i$ is defined inductively as follows: $v_0\equiv \Box \bot$ and
$v_n\equiv \bigwedge_{i:n_i=1} \Diamond v_i$ where by $n_i$ we
denote the $i$-th bit of $n$ when $n$ is written in binary and the
least significant bit is numbered 0. So, for example $v_1=\Diamond
v_0$, $v_2=\Diamond v_1$, $v_5=\Diamond v_2\land \Diamond v_0 =
(\Diamond\Diamond v_1) \land \Diamond v_0$ and so on.  Observe that
$v_0$ can only be true in a state with no successor states.  Also,
what is important is that these formulae allow us to encode very large numbers
using only a very small modality depth and no variables (or just one variable
if $\bot$ is considered short for $x\land\neg x$).

\item Next, we need to encode the literals of $\phi_p$. The modal formula
$\LL(x_i)$ is defined as $\LL(x_i)\equiv\Diamond v_i \land \Box v_i$. The
formula $\LL(\neg x_i)$ is defined as $\LL(\neg x_i)\equiv \Diamond v_i \land
\Diamond v_0 \land \Box \left( v_i \lor v_0 \right)$.

\item Now, to encode clauses we set $\CC(l_1\lor l_2\lor\ldots\lor l_k)\equiv
\left(\bigwedge_{i=1}^k \Diamond \LL(l_i)\right) \land \Box
\left(\bigvee_{i=1}^k \LL(l_i)\right)$.

\item Finally, to encode the whole formula we use $\FF(c_1\land
c_2\land\ldots\land c_m)\equiv\bigwedge_{i=1}^m \Diamond \CC(c_i)$

\end{itemize}

So far we have described how to construct a modal formula $\FF( \phi_p )$ from
$\phi_p$. $\FF(\phi_p)$ encodes the structure of $\phi_p$. Now we need to add
two more ingredients: we must describe with a modal formula that $\phi_p$ is
satisfied by an assignment and that the assignment is consistent among clauses.
We give two more formulae:

\begin{itemize}

\item $\SSS \equiv \Box\Diamond \left[ \left((\Diamond v_0) \to (\Box \neg
y)\right) \land \left( (\neg \Diamond v_0)\to (\Box y) \right) \right]$, where
we have introduced a single variable $y$.

\item $\CCC(n) \equiv \bigwedge_{i=1}^n \Diamond\Diamond\Diamond (y \land v_i)
\leftrightarrow \neg \Diamond\Diamond\Diamond (\neg y \land v_i)$

\end{itemize}

Our full construction is, given a propositional CNF formula $\phi_p$ with $n$
variables named $x_1,\ldots,x_n$, we create the modal formula
$\phi_m\equiv\FF(\phi_p)\land \SSS\land \CCC(n)$. 
%We need to prove several
%heplful lemmata.

\begin{lemma}

$\phi_p$ is satisfiable if and only if $\phi_m$ is satisfiable in K.

\end{lemma}

\begin{proof}

Suppose that $\phi_m$ is true in a state $s$ of some Kripke structure.  Then
$\CCC(n)$ is true in $s$ therefore for each $i$ we have either
$\Diamond\Diamond\Diamond (y\land v_i)$ is true in $s$ or
$\Diamond\Diamond\Diamond (\neg y\land v_i)$ is true in $s$. From this we
create a satisfying assignment: for those $i$ for which the first holds we set
$x_i=\top$ and for the rest $x_i=\bot$. We will show that this assignment
satisfies $\phi_p$.

Suppose that it does not, therefore there is some clause $c_i$ which is not
satisfied. However, since $\FF(\phi_p)$ is true in $s$ there exists a state $p$
with $sRp$ such that $\CC(c_i)$ is true in $p$. In every successor state of $p$
we have that $\LL(l_j)$ is true for some literal $l_j$ of $c_i$ and there
exists such a state for every literal of $c_i$. Also, in $s$ we have that
$\SSS$ is true, therefore in $p$ we have $\Diamond \left[ \left((\Diamond v_0)
\to (\Box \neg y)\right) \land \left( (\neg \Diamond v_0)\to (\Box y) \right)
\right]$. Therefore, in some $q$ such that $pRq$ we have $\left((\Diamond v_0)
\to (\Box \neg y)\right) \land \left( (\neg \Diamond v_0)\to (\Box y) \right)$
and we also have that $\LL(l_j)$ is true for some literal $l_j$ of $c_i$.
Suppose that $l_j$ is a negated literal, that is $l_j\equiv \neg x_k$. Then
$\LL(l_j) \equiv \Diamond v_k \land \Diamond v_0 \land \Box (v_k \lor v_0)$.
Therefore, since $\Diamond v_0$ is true in $q$ this means that $\Box \neg y$ is
true. Because $\Diamond v_k$ and $\Box \neg y$ are both true in $q$ there
exists an $r$ such that $qRr$ and $v_k\land \neg y$ is true in $r$. But then
$\Diamond\Diamond\Diamond (v_k\land \neg y)$ is true in $s$ which implies that
our assignment gives the value false to $x_k$. Since $c_i$ contains $\neg x_k$
it must be satisfied by our assignment, a contradiction. Similarly, if
$l_j\equiv x_k$ then $\LL(l_j)\equiv \Diamond v_k \land \Box v_k$. Clearly,
$v_0$ and $v_k$ cannot be true in the same state for $k>0$ therefore in $q$ we
have $\neg \Diamond v_0$ which implies $\Box y$. Therefore in some $r$ with
$qRr$ we have $y\land v_k$ which implies that our assignment sets $x_k$ to true
and since $c_i$ has the literal $x_k$ it must be satisfied.

The other direction is easier. First, we must construct for every $v_i$ a
Kripke structure to satisfy it. For $v_0$ this is a structure with just one
state with no successors. For $v_n$ we take the union of the structures for
every $v_i$ such that $n_i=1$. In this union for all $i$ such that $n_i=1$
there is a state for which $v_i$ is true, call it $s_i$. We add a state $s_n$
and set $s_nRs_i$ for all $i$ such that $n_i=1$. Clearly, $v_n$ is true in
$s_n$.

Now the construction of a Kripke structure for $\phi_m$ is straightforward. We
take the union of the structures for $v_i$, with $0\le i\le n$, thus we have a
state where $v_i$ is true for every $i$.  For every $i$ with $1\le i\le n$ we
create two more states: the first has as its only successor the state where
$v_i$ is true.  The other has two successors: the state where $v_i$ is true and
the state where $v_0$ is true. Thus, for each $i$ we have a state where
$\LL(x_i)$ is true and a state where $\LL(\neg x_i)$ is true. For every clause
we create a state and for each literal $l_j$ in the clause we add a transition
to the state where $\LL(l_j)$ is true. Therefore, for each clause $c_i$ we have
a state where $\CC(c_i)$ is true. Finally, we add a state and transitions to
all the states where some $\CC(c_i)$ is true. Clearly, $\FF(\phi_p)$ is true in
that state, which we call the root state. It is not hard to see that $\CCC(n)$
will also be satisfied independent of where $y$ is true, because for every $i$
we have made a unique state $p_i$ where $v_i$ is true and $p_i$ is at distance
exactly 3 from the root.

Take a satisfying assignment; for every $x_i$ which is true set the variable
$y$ to true in the states of the Kripke structure where $v_i$ is true. Set $y$
to false in every other state. Now, we must show that $\SSS$ is true in the
root state. This is not hard to verify because for every clause in the original
formula there is a true literal, call it $l$. If that literal is not negated
then in the state where $\LL(l)$ is true we have $\neg \Diamond v_0$ (because
the literal is not negated) and $\Box y$ (because the literal is true, so its
variable is true thus we must have set $y$ to true in the variable's
corresponding state). Therefore $(\neg \Diamond v_0 \to \Box y) \land (\Diamond
v_0 \to \Box \neg y)$ is true in the literal's corresponding state and
$\Diamond \left[ (\neg \Diamond v_0 \to \Box y) \land (\Diamond v_0 \to \Box
\neg y) \right]$ is true in the clause's corresponding state. Similar arguments
can be made for a negated literal. Since we start with a satisfying assignment
the same can be said for every clause, thus $\SSS$ is also true in the root
state. 

\end{proof}

%Next, we need to show that the produced modal formula has very small depth.

\begin{lemma} \label{lem:smdepth}

Suppose that $\phi_p$ is a propositional CNF formula with $n$ variables. Then,
if $tow(h)\ge n$ the formula $\phi_m \equiv \FF(\phi_p)\land \SSS\land \CCC(n)$
has modality depth at most $4+h$, where $tow(h)$ is the inductively defined
function $tow(0)=0$ and $tow(h+1)=2^{tow(h)}$.

\end{lemma}

\begin{proof}

First observe that the modality depth of $\phi_m$ is at most $3+\max_{0\le i\le
n} \md(v_i)$. Therefore, we just have to bound the modality depth of $v_i$.

We will use induction on $h$ to show that $tow(h)\ge n \Rightarrow \md(v_n)\le
h+1$.  For $h=0$ we have $tow(h)\ge n \Rightarrow n=0$, therefore $\md(v_0)=1$
and the proposition holds.

Suppose that the proposition holds for $h$.

Observe that $\md(v_n) \le 1 +
\max_{0\le i \le \log n}\{  \md(v_i) \}$ because writing $n$ in binary takes at
most $\log n + 1$ bits. If we have $n \le tow(h+1)$ then $\log n \le tow(h)$.
From the inductive hypothesis $\md(v_i) \le h+1$ for $i\le \log n$.  Therefore,
$\md(v_n) \le h+2$ and the proposition holds. 

\end{proof}

%Putting it all together, we get the following lower bound result:

\begin{theorem} \label{thm:lower}

There is no algorithm which can solve modal satisfiability in K for formulae
with a single variable and modality depth $d$ in time $f(d)\cdot poly(|\phi|)$
with $f(d) = O( tow(d-5))$, unless P=NP.

\end{theorem}

\begin{proof}

Suppose that there exists an algorithm A which in time $f(d)\cdot poly(|\phi|)$
can decide if a modal formula $\phi$ with modality depth $d$ and just one
variable is satisfiable. We will use this algorithm to solve propositional
satisfiability in polynomial time.

Given a propositional CNF formula $\phi_p$ we construct $\phi_m$ as described,
and if $\phi_p$ has $n$ variables let $H=\min\{h\ |\ n\le tow(h)\}$. Then
$\md(\phi_m)\le H+4$ and of course $\phi_m$ can be constructed in time
polynomial in $|\phi_p|$. Now we can use the hypothetical algorithm to see if
$\phi_m$ is satisfiable.

We have that $f(d)= O(tow(d-5))$. Therefore, running this algorithm will take
time $f(H+4)\cdot poly(|\phi_m|) = O(tow(H-1) \cdot poly(|\phi_m|))$.  But by
the definition of $H$ we have $tow(H-1)\le n$, therefore this bound is
polynomial in $|\phi_m|$ and therefore, also in $|\phi_p|$, which means that we
can solve an NP-complete problem in polynomial time. 

\end{proof}

\renewcommand{\dd}{\ensuremath \mathrm{d}_\Diamond}
\newcommand{\bd}{\ensuremath \mathrm{d}_\Box}

\section{Diamond Dimension} \label{sec:dimension}

In this Section we attempt to find some structural characteristics
of modal formulae which will allow us to beat the prohibitive
running time of modality depth. We define two measures, diamond
dimension and box dimension and show how they can be used to solve
satisfiability and validity respectively with a much lower running
time than modality depth.

\begin{definition}

Let $\phi$ be a modal formula in negation normal form, that is, with the $\neg$
symbol appearing only directly before propositional variables. Then its diamond
dimension, denoted by $\dd(\phi)$ is defined inductively as follows:

\begin{itemize}

\item $\dd(p)= \dd(\neg p) = 0$, if $p$ is a propositional letter

\item $\dd(\phi_1 \land \phi_2) = \dd(\phi_1) + \dd(\phi_2)$

\item $\dd(\phi_1 \lor \phi_2) = \max\{\dd(\phi_1),\dd(\phi_2) \}$

\item $\dd(\Box \phi) = \dd(\phi)$

\item $\dd(\Diamond \phi) = 1 + \dd(\phi)$

\end{itemize}

\end{definition}

%Before we go on let us give some intuition behind our definition. It is not
%hard to see that this definition is not as symmetrical as that of modality
%depth and also that the two properties of a formula that may cause it to have a
%high diamond dimension are high nesting of $\Diamond$s and conjuctions. The
%reasoning is that when we are trying to find a satisfying model for a formula
%we are forced to create a new state for every $\Diamond$ and when we see a
%conjunction we are forced to satisfy both of its parts. Therefore diamond
%dimension is meant to upper bound the size of the smallest satisfying model for
%a formula and we will see that it does exactly that.

For some intuition, observe that satisfiability becomes easy if we can somehow
place a small upper bound on the number of states needed in a satisfying model.
Our goal with this measure is to prove that if $\dd(\phi)$ is small then
$\phi$'s satisfiability can be checked in models with few states. This is why
the two properties of $\phi$ which can increase $\dd(\phi)$ are $\Diamond$
(which requires the creation of a new state) and $\land$ (which requires the
creation of states for both parts of the conjunction).

\begin{theorem}

If a modal formula $\phi$ is satisfiable and $\dd(\phi)\le k$ then there exists
a Kripke structure with $O(k!)$ states which satisfies $\phi$.

\end{theorem}

\begin{proof}

Suppose that there exists a Kripke structure which satisfies $\phi$, that is
there exists some state $s$ in that structure where $\phi$ holds. We will
construct a working set of modal formulae $S$ which will satisfy the following
properties:

\begin{itemize}

\item All formulae in $S$ hold in $s$.

\item ($\bigwedge_{\phi_i\in S} \phi_i) \to \phi$ is a valid formula.

\item $\dd(\phi) \ge \sum_{\phi_i\in S} \dd(\phi_i) $.

\end{itemize}

We begin with $S=\{\phi\}$ which obviously satisfies the above properties. We
will apply a series of transformations to $S$ while retaining these properties
until eventually we reach a point where every formula in $S$ is simple (in a
sense we will make precise later) and then we will construct a model with the
promised number of states for $\phi$.

While possible we apply the following rules to $S$:

\begin{enumerate}

\item If there exists a formula $\phi_i\in S$ such that $\phi_i = \phi_i^1
\land \phi_i^2$ then remove $\phi_i$ from $S$ and add $\phi_i^1$ and $\phi_i^2$
to $S$.

\item If there exists a formula $\phi_i\in S$ such that $\phi_i = \phi_i^1
\lor \phi_i^2$ then remove $\phi_i$ from $S$. If $\phi_i^1$ is true
in state $s$ add $\phi_i^1$ to $S$, otherwise add $\phi_i^2$ to $S$.

\item If there are two formulae $\phi_i = \Box \psi_i$ and $\phi_j=\Box \psi_j$
in $S$ then remove them and insert the formula $\Box (\psi_i \land \psi_j)$.

\end{enumerate}

It should be clear that rule one does maintain the properties of $S$. Rule two
also maintains the properties: property one is maintained because we assumed
that $\phi_i$ is true in state $S$ therefore if $\phi_i^1$ is not true we add
$\phi_i^2$ which must be true. The other properties are also straightforward.
Finally, rule three follows from the fact that $\Box \phi_1 \land \Box \phi_2
\leftrightarrow \Box (\phi_1 \land \phi_2)$ is a valid formula.

It should be clear that applying all the rules until none applies will take
polynomial time.  When we can no longer apply the rules we have that $S=\{\Box
\psi, \Diamond \phi_1, \ldots, \Diamond \phi_k, l_1, \ldots, l_m\}$, where the
$l_i$ are propositional literals; in other words, we have (at most) one formula
that starts with a $\Box$.

Now we will use induction on the diamond dimension to prove our theorem. Let
$s(d)$ be a function which upper bounds the number of states in the smallest
model which are needed to satisfy formulae of depth $d$ (we are going to
calculate $s(d)$ recursively and prove that it is finite).  First, we can say
that $s(0)=1$, because a formula with diamond dimension 0 has no diamonds.
Therefore, $S$ contains one formula that starts with a $\Box$ and some
literals, for which there exists an assignment to make them all true (because
of the first property of $S$). Clearly, a model with just one state where we
pick this assignment will also make the formula that starts with $\Box$
trivially true, and by the second property of $S$ will satisfy $\phi$.

For the inductive step, suppose that all the satisfiable formulae of dimension
at most $\dd(\phi)$ need at most $s(d)$ states to be satisfied, where $d$ is
the formula's dimension.  Let's consider the diamond dimension of all the
formulae in $S$.   There are three cases: either $S$ does not have a formula
that starts with a $\Box$, or it doesn't have any formulae that start with
$\Diamond$, or it has both.

Suppose that all the formulae in $S$ are literals or start with $\Diamond$. In
this case, we have for all $\phi_i$ that $\dd(\phi_i)<\dd(\phi)$. Using the
inductive hypothesis we get that the number of states to satisfy each formula
$\phi_i$ is at most $s(\dd(\phi_i))$. Clearly, we can create a model which is
the union of the models for all the $\phi_i$ plus one state where we give an
appropriate assignment to the literals and appropriate transitions so that
$\Diamond \phi_i$ is true for all $i$. This model has at most $1+\sum_{i=1}^k
s(\dd(\phi_i))$ states.

If we have no formulae starting with diamonds we can easily see that the same
model as in the base case suffices, since $\Box \psi$ is trivially true in a
state without successors. So in this case we have just one state.

Finally, if we have both types of formulae in $S$ we construct the
following model: consider all the formulae $\psi \land \phi_i$, for
all $i$. Clearly, they are satisfiable, because $\Box \psi \land
\Diamond \phi_i$ is true in $s$. We know from the third property of
$S$ that $\dd(\phi) \ge \dd(\psi) + k + \sum_{i=1}^k \dd(\phi_i)$.
Therefore, $\dd(\psi \land \phi_i) = \dd(\psi) + \dd(\phi_i) \le
\dd(\phi) -k - \sum_{j\neq i} \dd{\phi_j}\le \dd(\phi)-1$. Now, we
take the union of the models for each $\psi\land\phi_i$, and each
model has at most $s(d-1)$ states. We add one state and transitions
to the appropriate states where $\psi\land \phi_i$ are true, which
together with an appropriate assignment makes all formulae of $S$
true in that state. The number of states is at most $1+k\cdot
s(\dd(\phi)-1)$.

Using the simple fact that $k\le \dd(\phi)$ we get from the above that $s(d)$
is upper bounded by $s(d)\le 1+ d\cdot s(d-1)$ which gives that $s(d)=O(d!)$.

\end{proof}

\begin{corollary}

Given a modal formula $\phi$ with $v$ variables and diamond
dimension $\dd(\phi)=k$ we can solve the satisfiability problem for
$\phi$ in time $O(2^{O(k!)\cdot v}\cdot|\phi|)$.

\end{corollary}

\begin{proof}

It follows from the proof of the previous theorem that if $\phi$ is
satisfiable there exists a model of a specific type which can
satisfy it; specifically it can be satisfied in a model where the
states are connected in a tree where the root has $k$ children, each
of which has $k-1$ children, each of which has $k-2$ children and so
on. This tree has $O(k!)$ states and exhausting all possible truth
assignments to the variables in all the states and using the fact
that model checking can be performed in linear time we get the stated running
time.

\end{proof}

%It should be noted that the bound in the previous theorem is probably rather
%loose. We conjecture that with a better analysis it might be possible to show
%that the running time is in fact singly exponential on the diamond dimension.

Let us now tackle the validity problem. The most straightforward way
to check the validity of a formula is to check whether its negation
is satisfiable. Diamond dimension is not likely to help us directly
in this case because if a formula has low diamond dimension this
does not imply that its negation also has low dimension. Therefore,
we define a dual measure called box dimension.

\begin{definition}

Let $\phi$ be a modal formula in negation normal form. Then its box dimension,
denoted by $\bd(\phi)$ is defined inductively as follows:

\begin{itemize}

\item $\bd(p)= \bd(\neg p) = 0$, if $p$ is a propositional letter

\item $\bd(\phi_1 \lor \phi_2) = \bd(\phi_1) + \bd(\phi_2)$,
$\bd(\phi_1 \land \phi_2) = \max\{\bd(\phi_1),\bd(\phi_2) \}$

\item $\bd(\Box \phi) = 1+ \bd(\phi)$,  $\bd(\Diamond \phi) = \bd(\phi)$

\end{itemize}

\end{definition}

\begin{theorem}

For any formula $\phi$ we have $\dd(\phi) = \bd(\neg \phi)$.

\end{theorem}

\begin{proof}

We use induction on the length of the formula. For formulae which are just
propositional letters or literals it is obviously true. Now take a formula
$\phi$. If $\phi=\Box \psi$ then $\dd(\phi)=\dd(\psi)$. Also, $\neg \phi =
\Diamond \neg \psi$ and $\bd(\neg \phi) = \bd (\neg\psi)$. Thus, $\dd(\phi) =
\bd(\neg \phi)$ by the inductive hypothesis. The proof is similar in the other
cases.

\end{proof}

\begin{corollary}

Given a modal formula $\phi$ with $v$ variables and box dimension $\bd(\phi)=k$
we can solve the validity problem for $\phi$ in time $O(2^{O(k!)\cdot
v}\cdot|\phi|)$.

\end{corollary}

\newcommand{\mw}{\textrm{mw}}

\section{Modal Width} \label{sec:width}

In this section we give another structural parameter for modal
formulae called modal width in an attempt to solve modal
satisfiability even more efficiently. We will show that
satisfiability and validity can be solved in time only singly
exponential in the modal width and $v$.

First we define inductively the function $s(\phi)$ which given a modal formula
returns a set of modal formulae. Intuitively, whether $\phi$ holds in a given
state $s$ of a Kripke structure depends on two things: the values of the
propositional variables in $s$ and the truth values of some formulae $\psi_i$
in the successor states of $s$. These formulae are informally the subformulae
of $\phi$ which appear at modal depth 1. $s(\phi)$ gives us exactly this set of
formulae.

\begin{itemize}

\item $s(p)=\emptyset$ if $p$ is a propositional letter

\item $s(\neg \phi)=s(\phi)$, $s(\phi_1\lor \phi_2)=s(\phi_1\land
\phi_2) = s(\phi_1) \cup s(\phi_2)$

\item $s(\Box \psi) = s(\Diamond \psi) = \{ \psi \}$

\end{itemize}

Now we inductively define the set $S_i(\phi)$, which intuitively corresponds to
the set of subformulae of $\phi$ at depth $i$.

\begin{itemize}

\item $S_0(\phi) = s(\phi)$

\item $S_{i+1}(\phi) = \bigcup_{\psi\in S_i(\phi)} s(\psi)$

\end{itemize}

Finally, we can now define the modal width of a formula $\phi$ at depth $i$ as
$\mw_i(\phi)=|S_i(\phi)|$ and the modal width of a formula as
$\mw(\phi)=\max_i\mw_i(\phi)$.

Before we go on, let us prove a basic observation regarding $\mw_i(\phi)$ and
$\md(\phi)$.

\begin{lemma} 
For all $i\ge \md(\phi)$ we have $\mw_i(\phi)=0$.
\end{lemma}

\begin{proof}

Observe that for all formulae $\phi$ such that $\md(\phi)\ge 1$ we have
$\md(\phi) > \max_{\psi\in s(\phi)} \md(\psi)$. Using this fact the proof
follows easily by induction on $\md(\phi)$.

\end{proof}

%\begin{definition}
%
%Let $\phi$ be a modal formula. Then its modal width at level $i$,
%denoted by $\mw_i(\phi)$ is defined inductively as follows:
%
%\begin{itemize}
%
%\item $\mw_i(p)=0$ for all $i\ge 0$ if $p$ is a propositional letter
%
%\item $\mw_0(\Box \phi) = \mw_0(\Diamond \phi) = 1$
%
%\item $\mw_{i+1}(\Box \phi) = \mw_{i+1}(\Diamond \phi) =
%\mw_i(\phi)$
%
%\item $\mw_i(\phi_1 \lor \phi_2) = \mw_i(\phi_1 \land \phi_2) =
%\mw_i(\phi_1) + \mw_i(\phi_2)$ for all $i\ge 0$
%
%\item $\mw_i(\neg \phi) = \mw_i(\phi)$  for all $i\ge 0$
%
%
%\end{itemize}
%
%We define the modal width of the formula as $\mw(\phi) = \max_{i\ge
%0} \mw_i(\phi)$.
%
%\end{definition}
%
%To give some intuition behind this definition, we can think of the
%modal width as measuring how many different modal subformulae our
%formula contains at depth $i$. The idea is that the truth value of
%the subformulae of depth $i$ at some state $s$ depends only on the
%truth value of the subformulae of depth $i+1$ at the successors of
%$s$. If the maximum width of the formula is bounded we can
%exhaustively check all possible truth values for subformulae at the
%next level of depth and decide if some particular truth assignment
%to the subformulae of depth $i$ is possible.
%
%This intuition is made much more precise in the proof of the
%following theorem:

\begin{theorem}

There exists an algorithm which decides the satisfiability of a
modal formula $\phi$ with $v$ variables, $\md(\phi)=d$ and
$\mw(\phi)=w$ in time $O(2^{2v+3w}\cdot d\cdot w\cdot |\phi|)$.

\end{theorem}

\begin{proof}

%First we define inductively the function $s(\phi)$ which given a
%modal formula returns a set of modal formulae.
%
%\begin{itemize}
%
%\item $s(p)=\emptyset$ if $p$ is a propositional letter
%
%\item $s(\neg \phi)=s(\phi)$, $s(\phi_1\lor \phi_2)=s(\phi_1\land
%\phi_2) = s(\phi_1) \cup s(\phi_2)$
%
%\item $s(\Box \psi) = s(\Diamond \psi) = \{ \psi \}$
%
%\end{itemize}
%
%Then we define the set $S_i(\phi)$, which intuitively corresponds to
%the set of subformulae of $\phi$ at depth $i$.
%
%\begin{itemize}
%
%\item $S_0(\phi) = s(\phi)$
%
%\item $S_{i+1}(\phi) = \bigcup_{\psi\in S_i(\phi)} s(\psi)$
%
%\end{itemize}
%
%We make use of the following two claims whose proofs are omitted due to space constraints.
%
%\begin{itemize}
%
%\item For all $i$, $|S_i(\phi)|\le \mw_i(\phi)$
%
%\item For all $i\ge \md(\phi)$ we have $\mw_i(\phi)=0$.
%
%\end{itemize}

We will need to use a function $Prop(\phi)$ which, given a modal formula
$\phi$, returns a propositional formula which corresponds to $\phi$ with all
modal subformulae replaced by new propositional variables. $Prop(\phi)$ can be
inductively defined as follows:

\begin{itemize}

\item $Prop(p)=p$ if $p$ is a propositional letter.

\item $Prop(\phi_1\lor \phi_2) = Prop(\phi_1) \lor Prop(\phi_2)$,
  $Prop(\phi_1\land\phi_2) = Prop(\phi_1) \land Prop(\phi_2)$,
  $Prop(\neg \phi_1) = \neg Prop(\phi_1)$

\item $Prop(\Box \phi_1) = q_j$, where $q_j$ is
a new propositional letter.

\end{itemize}

Notice that once again we consider $\Diamond \phi$ as shorthand
for $\neg \Box \neg \phi$.

Let $P=\{p_1, p_2, \ldots, p_v\}$ be the set of propositional
variables appearing in $\phi$. For all $i\in \{0,\ldots,d-1\}$, for
all $P'\subseteq P$ and for all $S'\subseteq S_i(\phi)$ we define
the formula $F(i,P',S')=\left(\bigwedge_{p_i\in P'}p_i\right) \land
\left(\bigwedge_{p_i\in P\setminus P'} \neg p_i \right) \land
\left(\bigwedge_{\psi_i\in S'}\psi_i\right) \land
\left(\bigwedge_{\psi_i\in S_i(\phi)\setminus S'} \neg \psi_i
\right) $. Clearly there are at most $2^{v+w}d$ formulae
$F(i,P',S')$ defined and for each one of these we will compute
whether it is satisfiable or not using dynamic programming. We will
use a boolean matrix $A(i,P',S')$ of size $2^{v+w}d$ to store the results.

First, we have $S_d(\phi)=\emptyset$. It is not hard to see that all
formulae $F(d,P',\emptyset)$ are indeed satisfiable, so we initialize the
corresponding entries in $A$ to True. Suppose now that for some $i$ we have
filled out completely all entries $A(i+1,P',S')$. We will show how to fill out
any position in row $i$, say position $A(i,P',S')$. The crucial part now is that
if we consider the formula $Prop(F(i,P',S'))$, it will have some new variables
$q_i$ which correspond to modal subformulae which all appear in $S_{i+1}(\phi)$.

The formula $Prop(F(i,P',S'))$ has at most $v+w$ variables. It is not hard to see
that if $F(i,P',S')$ is satisfiable, then $Prop(F(i,P',S'))$ is also satisfiable,
so our first step is to check this. The truth assignments for the $v$ variables
are easy to infer, therefore we only need to go through the $2^w$ possible assignments
for the new variables. For each satisfying assignment we find we then need to
check if a model that satisfies $F(i,P',S')$ can be built from it.

So, suppose that $Q$ is the set of new variables, and we have found an assignment
which sets the variables of $Q'\subseteq Q$ to true and the rest to false and
satisfies $Prop(F(i,P',S'))$. Each variable $q_j$ of $Q$ corresponds to a formula
$\Box \phi_j$ with $\phi_j\in S_{i+1}(\phi)\cup P$. If $q_j\in Q'$ we must make sure that
$\phi_j$ is true in all successors of the state $s$ where $F(i,P',S')$ will hold, in
the model we are building. Let $S''\subseteq S_{i+1}(\phi)\cup P$ be the set of formulae $\phi_j$
which we conclude that must hold in all successors of $s$ in this way.

If $q_j\not\in Q'$ we have that $\neg \Box \phi_j$ must
hold in $s$, thus $s$ must have a successor where $\neg \phi_j$ is true, or equivalently
$\phi_j$ is false. Let $S^{*}\subseteq S_{i+1}(\phi)\cup P$ be the set of formulae $\phi_j$
for which we conclude that they must be false in some successor of $s$ in this way.

To decide if it is possible to build appropriate successors to $s$
so that all these conditions are satisfied, we look at row $i+1$ of $A$. Specifically
we consider the set of entries $A(i+1,P',S')$ such that $S''\subseteq S'\cup P'$ and
$A(i+1,P',S')=T$. Informally, these correspond to formulae which are satisfiable
(because the corresponding entry is set to true) and which also can serve as successors
to $s$ without violating the conditions of $S''$, that is, in any state where they hold
all formulae which we need to be true in all successors of $s$ are indeed true. Now, we
simply check if for each $q_j\in S^{*}$ there exists an entry in the set we have
selected so far with $q_j\not\in S'\cup P'$. If this is the case we can conclude that $F(i,P',Q')$
is satisfiable and set the corresponding entry of $A$ to True, otherwise we conclude that
no satisfying model can be built from the assignment we get from $Q$, even though
$Prop(F(i,P',S'))$ is satisfied. This whole process of computing $S''$ and $S^{*}$
and checking through row $i+1$ of $A$ can be performed in time $O(w\cdot 2^{v+w}|\phi|)$.

To decide if the initial formula $\phi$ is satisfiable, we compute $Prop(\phi)$ and
perform the same process: for every satisfying assignment of $Prop(\phi)$ we look at
corresponding entries of row $0$ of $A$ to see if a model for $\phi$ can be built. The
total time for this algorithm is $O(2^{3w+2v}wd|\phi|)$, because for each of the at most
$2^{v+w}d$ entries of $A$ we need to check through at most $2^w$ assignments and for each
we spend at most $O(w\cdot 2^{v+w}|\phi|)$.

\end{proof}

\section{Conclusions and Open Problems}

In this paper we defined and studied several modal formula
complexity measures and investigated how each can be used to attack
cases of modal satisfiability. Our results show that proving
fixed-parameter tractability is only a first step in such problems,
because the dependence on the parameters can vary significantly and
some parameters offer much better algorithmic footholds than others.

It is worthy of remark that the measures of formula complexity we have
discussed are not directly comparable; for example it is possible to construct
a formula with small modality depth and very high modal width, or vice-versa.
In this sense it is not possible to infer solely from our results which formula
complexity measure is the ``best'', since each corresponds to a different
family of modal formulae. However, our results can be seen as a first attempt
at drawing a complexity ``map'' for different modal formula parameters, looking
for areas where satisfiability becomes more or less tractable. This perspective
creates a nice connection between this work and for example the research area
of graph widths, where the complexity of model checking problems on graphs is
explored in different graph families depending on a graph complexity measure.
This is a well-developed area whose insights may be applicable and helpful in
the study of the problems of this paper. (For a summary of the current
complexity ``map'' for graph width parameters see Figure 8.1\ in \cite{DBLP:journals/eccc/Grohe07})

Possible future directions are the investigation of yet more natural formula
complexity measures and attempting to improve the running times or to show good
lower bounds for the already known measures. Finally, extending our results to
other modal logics, such as modal logics where Kripke structures are required
to be reflexive or transitive (e.g. S4) would be an interesting next step.

\bibliographystyle{abbrv} \bibliography{paper}

%\newpage
%\section*{Appendix}
%The proof of theorem \ref{thm:depth_upper} is presented here:

\end{document}